\documentclass{scrartcl}

\usepackage{soul}
\usepackage{url}
\usepackage[utf8]{inputenc}
\usepackage[small]{caption}
\usepackage{graphicx}
\usepackage{amsmath}
\usepackage{amsthm}
\usepackage{booktabs}
\usepackage{algorithm}
\usepackage{algorithmic}
\usepackage{eurosym}
\usepackage{array}
\usepackage{colortbl}

\usepackage{wrapfig}
\usepackage{booktabs}
\usepackage{caption}
\usepackage{subcaption}
\usepackage{amsmath,amsfonts,amssymb}
\usepackage{bbm}
\usepackage{stmaryrd}

\usepackage{tikz}
\usetikzlibrary{automata, positioning, calc, shapes, arrows, fit}

\usepackage{turnstile}
\usepackage{mathrsfs}
\usepackage{amsthm}
\usepackage{nicefrac}
\usepackage{natbib}
\usepackage{dsfont}

\usepackage{xcolor}

\definecolor{my-red}{HTML}{C62828}
\definecolor{my-red-light}{HTML}{E57373}
\definecolor{my-red-verylight}{HTML}{FFCDD2}
\definecolor{my-red-dark}{HTML}{B71C1C}

\definecolor{my-pink}{HTML}{EC407A}
\definecolor{my-pink-light}{HTML}{F48FB1}
\definecolor{my-pink-verylight}{HTML}{F8BBD0}
\definecolor{my-pink-dark}{HTML}{880E4F}

\definecolor{my-purple}{HTML}{8E24AA}
\definecolor{my-purple-light}{HTML}{BA68C8}
\definecolor{my-purple-verylight}{HTML}{e5cefc}
\definecolor{my-purple-dark}{HTML}{6A1B9A}

\definecolor{my-indigo}{HTML}{3949AB}
\definecolor{my-indigo-light}{HTML}{7986CB}
\definecolor{my-indigo-verylight}{HTML}{9FA8DA}
\definecolor{my-indigo-dark}{HTML}{1A237E}

\definecolor{my-blue}{HTML}{1E88E5}
\definecolor{my-blue-light}{HTML}{64B5F6}
\definecolor{my-blue-verylight}{HTML}{B3E5FC}
\definecolor{my-blue-dark}{HTML}{0D47A1}

\definecolor{my-cyan}{HTML}{00BCD4}
\definecolor{my-cyan-light}{HTML}{4DD0E1}
\definecolor{my-cyan-verylight}{HTML}{80DEEA}
\definecolor{my-cyan-dark}{HTML}{0097A7}

\definecolor{my-teal}{HTML}{009688}
\definecolor{my-teal-light}{HTML}{4DB6AC}
\definecolor{my-teal-verylight}{HTML}{B2DFDB}
\definecolor{my-teal-dark}{HTML}{00695C}

\definecolor{my-green}{HTML}{39ac39}
\definecolor{my-green-light}{HTML}{8cd98c}
\definecolor{my-green-verylight}{HTML}{b3e6b3}
\definecolor{my-green-dark}{HTML}{339933}

\definecolor{my-grass}{HTML}{689F38}
\definecolor{my-grass-light}{HTML}{8BC34A}
\definecolor{my-grass-verylight}{HTML}{AED581}
\definecolor{my-grass-dark}{HTML}{33691E}

\definecolor{my-lime}{HTML}{CDDC39}
\definecolor{my-lime-light}{HTML}{DCE775}
\definecolor{my-lime-verylight}{HTML}{E6EE9C}
\definecolor{my-lime-dark}{HTML}{AFB42B}

\definecolor{my-yellow}{HTML}{fffc29}
\definecolor{my-yellow-light}{HTML}{fffd7a}
\definecolor{my-yellow-verylight}{HTML}{fefdbb}
\definecolor{my-yellow-dark}{HTML}{FFD600}

\definecolor{my-orange}{HTML}{FF8F00}
\definecolor{my-orange-light}{HTML}{FFC107}
\definecolor{my-orange-verylight}{HTML}{ffe5a4}
\definecolor{my-orange-dark}{HTML}{FF6F00}

\definecolor{my-brown}{HTML}{6D4C41}
\definecolor{my-brown-light}{HTML}{795548}
\definecolor{my-brown-verylight}{HTML}{BCAAA4}
\definecolor{my-brown-dark}{HTML}{3E2723}

\definecolor{my-gray}{HTML}{616161}
\definecolor{my-gray-light}{HTML}{9E9E9E}
\definecolor{my-gray-verylight}{HTML}{f0f0f0}
\definecolor{my-gray-dark}{HTML}{424242}

\definecolor{my-steel}{HTML}{546E7A}
\definecolor{my-steel-light}{HTML}{78909C}
\definecolor{my-steel-verylight}{HTML}{B0BEC5}
\definecolor{my-steel-dark}{HTML}{37474F}

\usepackage{pifont}
\newcommand{\cmark}{{\color{my-green}\ding{51}}}
\newcommand{\xmark}{{\color{my-red}\ding{55}}}

\usepackage[unicode=true, bookmarks=false, breaklinks=true, pdfborder={0 0 1}, colorlinks=true]{hyperref}
\hypersetup{linkcolor=my-blue-dark, citecolor=my-blue-dark, filecolor=my-blue, urlcolor=my-blue-dark}

\newtheorem{theorem}{Theorem}

\newtheorem{corollary}[theorem]{Corollary}

\newtheorem{proposition}[theorem]{Proposition}
\newtheorem{definition}{Definition}

\newtheorem{myex}{Example}

\newcommand{\tuple}[1]{\left\langle #1 \right\rangle}

\DeclareMathOperator*{\argmax}{argmax}
\DeclareMathOperator*{\argmin}{argmin}

\renewcommand{\phi}{\varphi}
\renewcommand{\vec}[1]{\boldsymbol{#1}}

\DeclareMathOperator{\symmdiff}{\triangle}

\newcommand{\complexP}{\ensuremath{\mathtt{P}}}
\newcommand{\complexNP}{\ensuremath{\mathtt{NP}}}

\newcommand{\agentSet}{\mathcal{N}}

\newcommand{\projSet}{\mathcal{P}}
\newcommand{\conceivProjSet}{\mathbb{P}}
\newcommand{\projSetAll}{\mathbb{P}}
\newcommand{\awareStruct}{\boldsymbol{\mathcal{C}}}

\newcommand{\allocSet}{{\normalfont\textsc{Feas}}}
\newcommand{\allocSetEx}{\allocSet_{\textsc{Ex}}}

\newcommand{\profile}{\boldsymbol{A}}
\newcommand{\shortProfile}{\boldsymbol{P}}
\newcommand{\idealSet}{\mathit{top}}
\newcommand{\topProfile}{\boldsymbol{top}}
\newcommand{\undom}{\mathit{undom}}

\newcommand{\pbRule}{{\normalfont\textsc{F}}}
\newcommand{\maxWelCard}{{\normalfont\textsc{MaxCard}}}
\newcommand{\maxWelCost}{{\normalfont\textsc{MaxCost}}}
\newcommand{\greedy}{{\normalfont\textsc{Greed}}}
\newcommand{\greWelCard}{{\normalfont\textsc{GreedCard}}}
\newcommand{\greWelCost}{{\normalfont\textsc{GreedCost}}}
\newcommand{\mes}{{\normalfont\textsc{MES}}}
\newcommand{\mesSat}[1]{\mes{}{\normalfont[#1]}}
\newcommand{\seqPhragmen}{{\normalfont\textsc{SeqPhrag}}}
\newcommand{\maximinSupport}{{\normalfont\textsc{MaximinSupp}}}

\newcommand{\tieRule}{\textsc{Untie}}
\newcommand{\canonTieRule}{\textsc{CanonUntie}}

\newcommand{\shortRule}{{\normalfont\textsc{Short}}}
\newcommand{\reprShortRule}{{\normalfont\textsc{ReprShort}}}
\newcommand{\nominationShortRule}{{\normalfont\textsc{NomShort}}}
\newcommand{\medianShortRule}{{\normalfont\textsc{MedianShort}}}

\newcommand{\satisfaction}{\mathit{sat}}

\newcommand{\cardSatisfaction}{\satisfaction^{\mathit{card}}}
\newcommand{\costSatisfaction}{\satisfaction^{\mathit{cost}}}

\title{Shortlisting Rules and Incentives in \\\ an End-to-End Model for \\ Participatory Budgeting\thanks{This paper is an extended version of the paper published in the proceedings of the International Joint Conference on Artificial Intelligence (IJCAI-21) \citep{REH21}.}}

\author{
	Simon Rey, Ulle Endriss \text{\normalfont and} Ronald de Hann\\[0.65em]
	\texttt{s.j.rey@uva.nl}, u.\texttt{endriss@uva.nl} \text{\normalfont and} \texttt{r.dehaan@uva.nl}
}

\date{Institute for Logic, Language and Computation (ILLC)\\[0.25em]University of Amsterdam}

\begin{document}
	\maketitle

	\begin{abstract}
		We introduce an end-to-end model for participatory budgeting
		grounded in social choice theory. 
		Our model accounts for the interplay between the two stages commonly encountered in real-life participatory budgeting. In the first stage participants propose projects to be shortlisted, while in the second stage they vote on which of the shortlisted projects should be funded. 
		Prior work of a formal nature has focused on analysing the second stage only.
		We introduce several shortlisting rules for the first stage 
		and analyse them in both normative and algorithmic terms. Our main focus is on the incentives of participants to engage in strategic behaviour during the first stage, in which they need to reason about how their proposals will impact the range of strategies available to everyone in the second stage.
	\end{abstract}
	
	\section{Introduction}
	
	Participatory budgeting (PB) is a loosely defined range of mechanisms designed to improve the involvement of ordinary citizens in public spending decisions \citep{Caba04, WNT21}. Though it is not easy to come up with a unified definition, a typical PB process is roughly organised around the following three stages \citep{Wamp00, Caba04, Shah07}:
	\begin{itemize}
		\item Recovering all the projects that can be potentially implemented through the PB process, and selecting a shortlist of them that will carry over to the next step;
		\item Collecting the opinions of the participants regarding the shortlisted projects in order to determine the ones that will actually be implemented;
		\item Monitoring the actual implementation of the projects and reporting to the participants about it.
	\end{itemize}
	From a social choice perspective, the second stage is the most exciting---how to collect and aggregate opinions is actually the research question at the very core of the discipline. It is thus probably not surprising that it has been the focus of almost all the social choice literature on PB \citep{AzSh21, ReMa23}.
	In this paper, we instead propose an \emph{end-to-end model} of PB that accounts for both the first and the second stages. By studying this model, we aim at better understanding real-life processes.
	
	At first glance, in may not be clear that the first stage is a social choice problem on its own. In real-life implementations of PB processes, it is usual that residents actually submit proposals directly to the municipality \citep{Shah07, WNT21}. Later on, the municipality selects a shortlist of the proposals by taking into account which proposals have been submitted (and potentially more information from the citizens). This is thus a perfect example of a social choice problem. Our end-to-end model will thus formalise the first stage of PB, where proposals are collected and shortlisted, and the second stage where the final budget allocation is selected. 
	
	\medskip
	
	Our contribution, beyond the formulation of the model itself, is twofold. First, we propose and analyse several shortlisting rules for the first stage. Second, we analyse the incentives of engaging in strategic manipulation when making proposals during the first stage---in view of how these affect the second stage. Let us briefly discuss both contributions.
	
	Our first point of focus in this new model will be the shortlisting stage. Interpreting it as multi-winner voting scenario---where the proposals submitted by the agents represent the ballots they cast, and the set of shortlisted proposals is the outcome---we will discuss different ways of determining the shortlist, different \emph{shortlisting rules}. It is important to keep in mind that there are no formal constraints on the outcome of the shortlisting stage. Indeed, any subset of proposals is an admissible outcome (except for the empty set maybe). Because of this, there is a lot of room to develop shortlisting rules. This raises the question of what makes some shortlists more desirable than others. By investigating what happens in practice, we identify four distinct objectives.
	\begin{enumerate}
		\item A first round of review usually removes the proposals that are simply infeasible, typically because of legal issues.\footnote{For instance, in the 2017 PB process in Paris, one of the proposals was to demolish the Sacré-C\oe ur, a church in the centre of Paris. This proposal was rejected by the municipality of Paris because it falls outside of its jurisdiction. For more details, see \href{https://www.francetvinfo.fr/france/ile-de-france/paris/affreux-disproportionne-un-parisien-propose-a-la-mairie-de-raser-le-sacre-coeur_2068737.html}{francetvinfo.fr/france/ile-de-france/paris/affreux-\linebreak{}disproportionne-un-parisien-propose-a-la-mairie-de-raser-le-sacre-coeur\_2068737.html} (in French).} This specific objective will not be incorporated into our analysis as we assume that all projects considered are implementable.
		\item Another goal of the shortlisting stage is to reduce the number of proposals entering the second stage. For instance, if we look at the PB exercises in Lisbon around 30\% of the projects were shortlisted \citep{AlAn14} while it reached 10\% in Toronto \citep{Murr19}.
		\item What we call the shortlisting stage often takes the form of rounds of public meetings where the proposals are discussed. During these meetings, citizens will typically develop their proposals, helped by other citizens and employees of the municipality. It is common for official organisers to avoid that projects that are too similar get proposed.\footnote{This has, for instance, been witnessed at PB meetings in Amsterdam.} This constitutes our third objective of the shortlisting stage: avoiding redundancy in the shortlist.
		\item One of the core objectives of PB is to provide citizens with a platform to express their opinion \citep{Wamp12}, and to witness direct impact of the democratic process. It is necessarily not possible to guarantee every citizen to be satisfied with, at least part of, the outcome of the voting stage---essentially because there is only a limited amount of money available. However, when shortlisting projects, there is no hard constraint on what can be selected. It is thus easier to ensure that everyone will have an impact on the shortlist. This is our fourth objective.
	\end{enumerate}
	In the first half of the paper, we will introduce different shortlisting rules and axioms to study them, motivated by the above observations. Specifically, inspired by the Thiele rules for committee elections \citep{Jans16}, we study \emph{diversity with respect to proposers}, by minimising the number of participants without shortlisted proposals, thus implementing objective number 4. Secondly, inspired by clustering algorithms \citep{JaDu88}, we also explore \emph{diversity w.r.t.\ the shortlist}, by avoiding selecting too many similar projects. This is an implementation of objective number 3. We also formalise some of the objectives in axiomatic terms in order to study the merits of the shortlisting rules.
	
	Our second point of focus concerns the interaction between the two stages of our model. More specifically, we will investigate the incentives for the agents to submit, or not, proposals during the shortlisting stage, because of the potential impact on the final budget allocation. The typical example of strategic behaviour here would be for an agent not to propose to build a fountain on the main square because another similar project is also proposed by someone else and the votes would be split between the two projects, preventing any of them to be selected eventually. To this end, we introduce the notion of \emph{first-stage strategy-proofness} and analyse how it depends on both the information available to participants and the choice of aggregation rules used during the two stages. 
	
	\paragraph{Related work.} 
	Most prior research of a formal nature regarding 
	PB has focused on the allocation stage only \citep{ReMa23}.
	Topics are ranging from concrete rules \citep{TaFa19, REH20}, to strategic behaviour \citep{GKSA19}, or proportionality \citep{ALT18, PPS21, BFLMP23}.
	Much of this work takes inspiration from the literature on multi-winner voting \citep{FSST17},
	exploiting the fact that electing a committee of $k$ representatives is isomorphic to selecting projects when each project costs 1~dollar and the budget limit is $k$~dollars. 
	
	We are not aware of any formal work regarding the shortlisting stage in PB. Having said this, also for this aspect of PB there are certain parallels to multi-winner voting. Note that in multi-winner voting the term `shortlisting' is used in two different senses: either to emphasise that choosing a set of $k$ candidates is but a first step in making a final decision \citep{FSST17}, or to refer to the problem of electing a set of variable size \citep{Kilg16, FST20, LaMa20}. Only the latter is formally related to shortlisting for PB
	(but does not involve costs or a budget limit). Note that its connections to clustering techniques have already been developed \citep{LaMa20}.
	
	\paragraph{Paper outline.} We develop our end-to-end model for PB in Section~\ref{sec:model}. Then, Section~\ref{sec:shortRules} is dedicated to shortlisting rules. We present a complete example in Section~\ref{sec:initialEx}. Shortlisting rules are studied axiomatically in Section~\ref{sec:AxiomsShortRule}. Finally, Section~\ref{sec:FSSP} is devoted to first-stage strategy-proofness.	
	
	\section{The Model}
	\label{sec:model}
	
	In this section, we introduce the two stages of our model for PB and fix our assumptions regarding agent preferences. 
	
	\subsection{Basic Notation and Terminology}
	
	We will denote by $\conceivProjSet = \{p_1, \ldots, p_m\}$ the (finite) set of all \emph{conceivable} projects. Those are projects that can be submitted in the first stage. The cost function $c: \conceivProjSet \to \mathbb{R}_{> 0}$ maps any conceivable project $p \in \conceivProjSet$ to its cost $c(p) \in \mathbb{R}_{> 0}$. The total cost of any set $P \subseteq \projSetAll$ is written $c(P) = \sum_{p \in P} c(p)$.
	We use $b \in \mathbb{R}_{> 0}$ to denote the budget limit and assume without loss of generality that for every project $p \in \conceivProjSet$, we have $c(p) \leq b$. The set of agents participating in the PB exercise is denoted by $\agentSet = \{1, \ldots, n\}$.
	
	\medskip
	
	We shall make use of the following generic procedure to choose a ``best'' subset (fitting the budget) of a given set of projects in view of a given ranking of those projects.
	
	\begin{definition}[Greedy Scheme]
		Consider a subset of projects $P \subset \conceivProjSet$, the cost function $c$, the budget limit $b$, and a strict ordering $\rhd$ over $\projSet$. The \emph{greedy scheme} $\greedy(P, c, b, \rhd)$ is a procedure selecting a budget allocation $\pi$ iteratively as follows. The budget allocation $\pi$ is initially empty. Projects are considered in the order defined by $\rhd$.
		When considering project $p$ for current budget allocation $\pi$, $p$ is selected (added to $\pi$) if and only $c(\pi \cup \{p\}) \leq b$.
		If there is a next project according to $\rhd$, it is considered; otherwise $\pi$ is the output of $\greedy(P, c, b, \rhd)$.
	\end{definition}
	
	We are going to require means for breaking ties, both between alternative projects and between alternative sets of projects.
	We will do so using tie-breaking rules. Those are functions---typically denoted by $\tieRule$---taking as input a set of tied projects, and returning a single one of them.
	
	The most natural way of breaking ties, when no other information is available, is to do so through the index of the projects. This is captured by the \emph{canonical tie-breaking rule}~$\canonTieRule$,  defined for any $P \subseteq \conceivProjSet$ as:
	\[\canonTieRule(P) = \argmin_{p_j \in P} j.\]
	
	We also use tie-breaking rules to transform weak orders over projects into strict orders. Let $\tieRule$ be an arbitrary tie-breaking rule. Take any weak order~$\succsim$ on $\conceivProjSet$. Then for every indifference class $P \subseteq \conceivProjSet$ of $\succsim$, we break ties as follows: $p = \tieRule(P)$ is the first project, then comes $\tieRule(P \setminus \{p\})$, and so forth. Overloading notation, we denote by $\tieRule(\succsim)$ the strict order thus obtained.
	
	Finally, we extend the canonical tie-breaking rule $\canonTieRule$ to non-empty sets $\mathfrak{P} \subseteq 2^\projSet$ of sets of projects in a lexicographic manner: $\canonTieRule(\mathfrak{P})$ is the unique set $P\in\mathfrak{P}$ such that $\canonTieRule(P \symmdiff P') \in P$ for all $P'\in\mathfrak{P}\setminus\{P\}$.\footnote{$\symmdiff$ is the symmetric difference between sets, defined for any $S$ and $S'$ as $S \symmdiff S' = (S \setminus S') \cup (S' \setminus S)$.} Thus, we require that, amongst all the projects on which $P$ and $P'$ differ, the one with the lowest index must belong to $P$. This formalises the way the usual lexicographic tie-breaking is lifted to sets of words (as in a dictionary for instance).
	
	When no tie-breaking rule is specified, it is assumed that ties are broken with respect to $\canonTieRule$.
	
	\subsection{The Shortlisting Stage}
	
	In the first stage, agents are asked to propose projects. 
	A \emph{shortlisting instance} is a tuple $\tuple{\conceivProjSet, c, b}$.
	Because of bounded rationality, an agent may not be able to conceive of all the projects they would approve of if only she were aware of them. We denote by $C_i \subseteq \conceivProjSet$ the set of projects that agent~$i$ can conceive of---their \emph{awareness set}---and we call the vector $\awareStruct = (C_1, \ldots, C_n)$ the \emph{awareness profile}. We do assume that agent~$i$ knows the cost of the projects in~$C_i$ as well as the budget limit~$b$.
	
	We denote by $P_i \subseteq C_i$ the set of projects agent~$i \in \agentSet$ chooses to actually propose, and we call the resulting vector $\shortProfile = (P_1, \ldots, P_n)$ a \emph{shortlisting profile}. We use $(\shortProfile_{-i}, P_i')$ to denote the profile we obtain when, starting from profile~$\shortProfile$, agent~$i$ changes their proposal to $P_i'$.
	
	A \emph{shortlisting rule}~$\shortRule$ maps any given shortlisting instance $I = \tuple{\conceivProjSet, c, b}$ and shortlisting profile $\shortProfile$ to a shortlist, \textit{i.e.}, a set $\shortRule(I, \shortProfile) \subseteq \bigcup \shortProfile$ of shortlisted projects, where $\bigcup \shortProfile = P_1 \cup \cdots \cup P_n$.
	
	\subsection{The Allocation Stage}
	
	In the second stage, agents vote on the shortlisted projects to decide which ones should get funded. An \emph{allocation instance} is a tuple $\tuple{\projSet, c, b}$, where $\projSet \subseteq \conceivProjSet$ is the set of shortlisted projects. Contrary to the shortlisting stage, the agents now know about all the projects they can vote for. They vote by submitting approval ballots, denoted by $A_i \subseteq \projSet$ for each $i \in \agentSet$, giving rise to a \emph{profile} $\profile = (A_1, \ldots, A_n)$. 
	The \emph{approval score} of a project~$p$ in profile~$\profile$ is $n_p^{\profile} = |\{i \in \agentSet \mid p \in A_i\}|$.
	
	For a given instance $I$, the goal is to choose a \emph{budget allocation} $\pi \subseteq \projSet$. It is feasible if $c(\pi) \leq b$, and $\allocSet(I)$ is the set of feasible budget allocations. Moreover, $\pi \in \allocSet(I)$ is \emph{exhaustive} if there exists no $p \in \projSet \setminus A$ such that $c(\pi \cup \{p\}) \leq b$, and $\allocSetEx(I)$ is the set of exhaustive budget allocations in $I$.
	
	An \emph{allocation rule} $\pbRule$ maps any given allocation instance~$I$ and profile~$\profile$ to a feasible budget allocation $\pbRule(I, \profile) \in \allocSet(I)$. An allocation rule~$\pbRule$ is \emph{exhaustive} if, for all instances~$I$ and all profiles~$\profile$, we have $\pbRule(I, \profile) \in \allocSetEx(I)$.
	
	We now introduce some allocation rules, as defined by \citet{ReMa23}.
	\begin{itemize}
		\item \greWelCost{} is defined for every instance $I = \tuple{\projSet, c, b}$ and profile $\profile$ as:
		\[\greWelCost(I, \profile) = \canonTieRule\left(\greedy(\projSet, c, b, \geq_{\mathit{app}}^{\profile})\right),\]
		where $\geq_{\mathit{app}}^{\profile}$ is such that $p \geq_{\mathit{app}}^{\profile} p'$ holds if and only if $n_{p}^{\profile} \geq n_{p'}^{\profile}$.
		\item \maxWelCost{} is defined for every instance $I = \tuple{\projSet, c, b}$ and profile $\profile$ as:
		\[\maxWelCost(I, \profile) = \canonTieRule\left(\argmax_{\pi \in \allocSet(I)} \sum_{A \in \profile} c(A \cap \pi)\right).\]
		
		\item \greWelCard{} is defined for every instance $I = \tuple{\projSet, c, b}$ and profile $\profile$ as:
		\[\greWelCard(I, \profile) = \canonTieRule\left(\greedy(\projSet, c, b, \geq_{\nicefrac{\mathit{app}}{c}}^{\profile})\right),\]
		where $\geq_{\mathit{app}}^{\profile}$ is such that $p \geq_{\nicefrac{\mathit{app}}{c}}^{\profile} p'$ holds if and only if $\nicefrac{n_{p}^{\profile}}{c} \geq \nicefrac{n_{p'}^{\profile}}{c}$.
		\item \maxWelCard{} is defined for every instance $I = \tuple{\projSet, c, b}$ and profile $\profile$ as:
		\[\maxWelCard(I, \profile) = \canonTieRule\left(\argmax_{\pi \in \allocSet(I)} \sum_{A \in \profile} |A \cap \pi|\right).\]
		
		\item \seqPhragmen{} constructs, for every instance $I = \tuple{\projSet, c, b}$ and profile $\profile$, budget allocations using the following continuous process.	Voters receive money in a virtual currency. They all start with a budget of~0 and that budget continuously increases as time passes. At time~$t$, a voter will have received an amount $t$ of money. For any time $t$, let $P^\star_t$ be the set of projects $p \in \projSet$ for which the supporters of $p$ altogether have more than $c(p)$ money available. As soon as, for a given $t$, $P^\star_t$ is non-empty, if there exists a $p \in P^\star$ such that $c(\pi \cup \{p\}) > b$, the process stops; otherwise one project from $P^\star_t$ is selected (using \canonTieRule{} if needed), the budget of its supporters is set to 0, and the process resumes.
		
		\item \maximinSupport{} behaves similarly as \seqPhragmen{} but allows for payments to be redefined in each round. We do not provide the actual definition but refer the reader to \citet{ReMa23} for it.s
		
		\item \mesSat{$\satisfaction$} is another rule similar to \seqPhragmen{}, parameterised by a satisfaction function $\satisfaction$. Roughly speaking, this rule behaves similarly as \seqPhragmen{} except that agents receive the money upfront. Once again we do not provide the full definition and refer the reader to \citet{ReMa23} for all the details. We denote by \mesSat{$\cardSatisfaction$} the method of equal shares used with the cardinality satisfaction function, and by \mesSat{$\costSatisfaction$} the method of equal shares used with the cost satisfaction function.
	\end{itemize}
	Note that these allocation rules are only used in Proposition~\ref{prop:NomShortRule_UFSSPA}.
	
	\subsection{Agent Preferences}
	
	Later on, we will discuss the incentives of the agents in our end-to-end model. To do so, we need ways of discussing their preferences. In the following we present what we assume to be the internal preference model that the agents follow.
	
	Consider a shortlisting instance $\tuple{\conceivProjSet, c, b}$.
	We make the assumption that agent~$i \in \agentSet$ has preferences over all individual projects in~$\conceivProjSet$---including those they are unaware of---and that those preferences take the form of 
	a strict linear order~$\rhd_i$ over $\conceivProjSet$. It is important to keep in mind that we do not assume that $i$ is aware of $\rhd_i$ in full.
	For any subset of projects $P \subseteq \conceivProjSet$, we denote by $\rhd_{i|P}$ the restriction of~$\rhd_i$ to~$P$. Our second assumption is that for any subset of projects $P \subseteq \conceivProjSet$, agent $i$ is able to determine an \emph{ideal set} of projects, denoted $\idealSet_i(P)$, that is determined through the use of the greedy selection procedure. Formally, we define:
	\[\idealSet_i(P) = \greedy(P, c, b, \rhd_{i|P}).\]
	This approach permits us to model what constitutes a \emph{truthful} vote by an agent for varying shortlists~$P$.  We call the vector $\topProfile(P) = (\idealSet_1(P), \ldots, \idealSet_n(P))$ the \emph{ideal profile} given~$P$.
	
	When investigating the potential strategic behaviour of the agents, we will need to compare different budget allocations from the perspective of the agents. We assume that agents derive preferences over budget allocations from their ideal sets through the use of \emph{completion principles} \citep{LaXi16}. A completion principle is a method that, given an ideal point (or top element), generates a weak order over subsets of projects.\footnote{One could also work with partial orders here. All of our definitions would carry over seamlessly in this case. Their interpretation would however differ.} For any ideal set $\idealSet \subseteq \conceivProjSet$, we denote by $\succsim_\idealSet$ the weak order over subsets of projects induced by a given completion principle (we omit the latter to simplify the notation; it will always be clear from the context). We will denote by $\succ_\idealSet$ the strict part of $\succsim_\idealSet$, and $\sim_\idealSet$ its indifference part. For instance, if we follow the \emph{cardinality-based} completion principle, then, given an ideal set $\idealSet \subseteq \conceivProjSet$, we have $P \succsim_\idealSet P'$ for every two subsets of projects $P, P' \subseteq \conceivProjSet$ such that $|P \cap \idealSet| \geq |P' \cap \idealSet|$. Under the \emph{cost-based} completion principle, we have $P \succsim_\idealSet P'$ for every two subsets of projects $P, P' \subseteq \conceivProjSet$ such that $c(P \cap \idealSet) \geq c(P' \cap \idealSet)$.
	
	Instead of stating our results for specific completion principles, we will phrase them so that they apply to all completion principles behaving in certain ways. In the following we introduce the different properties we will need. A completion principle generating $\succsim_\idealSet$ from $\idealSet$ is said to satisfy:
	
	\begin{itemize}
		\item \textbf{Top-First} if the ideal point $\idealSet$ strictly dominates any other subset of projects: for all $P \subseteq \projSet$, we have $\idealSet \succ_\idealSet P$;
		\item \textbf{Top-Sufficiency} if the empty set is strictly dominated by any non-empty subset of $\idealSet$: for all $P \subseteq \idealSet$, if $P \neq \emptyset$, then we have $P \succ_\idealSet \emptyset$;
		\item \textbf{Top-Necessity} if any subset of projects that does not intersect with $\idealSet$ is treated the same way as the empty set: for all $P \subseteq \conceivProjSet$, if $P \cap \idealSet = \emptyset$, then we have $P \sim_\idealSet \emptyset$;
		\item \textbf{Cost-Neutral Monotonicity} if selecting more projects from $\idealSet$ is strictly better than fewer, as long as they all have the same cost: for all $P, P' \subseteq \conceivProjSet$ such that $P \symmdiff P' \subseteq \idealSet$, if $|P \cap \idealSet| > |P' \cap \idealSet|$ and $c(p) = c(p')$ for any two projects $p, p' \in P \symmdiff P'$, then we must have $P \succ_\idealSet P'$.
	\end{itemize}
	So, a completion principle is top-first if $\idealSet$ is indeed the best outcome. It is top-sufficient if it is sufficient to have some projects from $\idealSet$ to be better than the empty set. It is top-necessary if it is necessary to have some projects from $\idealSet$ to be better than the empty set. Finally, it is cost-neutral monotonic if having more projects from $\idealSet$ is better than having less, even if those are different projects, provided that they all have the same cost.
	
	As a warm-up, the reader can check that both the cardinality- and the cost-based completion principles satisfy all of the above properties.
	
	We finally provide our last definition (for this section). For any weak order $\succsim_\idealSet$ and any family of subsets of projects $\mathfrak{P} \subseteq 2^\projSet$, we use $\undom(\succsim_\idealSet,\mathfrak{P})$ to denote the set of subsets of projects that are undominated in $\mathfrak{P}$ according to $\succsim_\idealSet$.
	
	\section{Shortlisting Rules}
	\label{sec:shortRules}
	
	Many allocation rules have been defined in the literature \citep{ReMa23}. This is however not the case for shortlisting rules. In this section, we therefore propose several of them.
	
	\medskip
	
	Our first shortlisting rule is what arguably is the simplest of them, the \emph{nomination (shortlisting) rule}. Following this rule, every agent acts as a nominator, \textit{i.e.}, someone whose proposals are always all accepted.
	
	\begin{definition}[Nomination Rule]
		\label{def:nomShortRule}
		The \emph{nomination rule} $\nominationShortRule$ returns, for every shortlisting instance $I = \tuple{\conceivProjSet, c, b}$ and shortlisting profile $\shortProfile$, the shortlist:
		\[\nominationShortRule(I, \shortProfile) = \bigcup \shortProfile.\]
	\end{definition}
	
	Although very natural, the nomination shortlisting rule is not effective in reducing the number of projects. So let us go through some more examples of shortlisting rules.
	
	\subsection{The Equal Representation Shortlisting Rule}
	
	Since the budget limit is not a hard constraint at the shortlisting stage, one of its objectives could be to ensure that every participant has a say in the decision. Building on this idea, we introduce the \emph{$k$-equal representation shortlisting rule}, a Thiele rule\footnote{Thiele rules are multi-winner voting rules that have received substantial attention \citep{LaSk23}. We briefly sketch their definition here. Let $\vec{w} = (w_1, w_2, \ldots)$ be an infinite weight vector. Assume we are aiming at selecting exactly $k \in \mathbb{N}_{> 0}$ projects. Given a shortlisting instance $I = \tuple{\conceivProjSet, c, b}$ and a shortlisting profile $\shortProfile = (P_1, \ldots, P_n)$, the $\vec{w}$-Thiele method is a multi-winner voting rule that selects $k$-sized subsets of projects $P$ with maximum weight, where the weight of $P$ is defined as $\sum_{i \in \agentSet} \sum_{j = 1}^{|P \cap P_i|} w_j$.} that maximises the minimum number of selected projects per agent. Here the parameter $k$ determines the maximum cost of the shortlist selected by the rule, expressed as a multiplier of $b$.
	
	\begin{definition}[$k$-Equal Representation Shortlisting Rule]
		\label{def:equalReprShortRule}
		Let $k \in \mathbb{N}$. The \emph{$k$-equal representation shortlisting rule} $\reprShortRule_k$ is defined for every shortlisting instance $I = \tuple{\conceivProjSet, c, b}$ and every shortlisting profile $\shortProfile = (P_1, \ldots, P_n)$ as:
		\[\reprShortRule_k(I, \shortProfile) = \canonTieRule \left(\argmax_{\substack{P \subseteq \bigcup \shortProfile\\c(P) \leq k \cdot b}} \sum_{i \in \agentSet} \sum_{\ell = 0}^{|P_i \cap P|} \frac{1}{n^\ell}\right).\]
	\end{definition}
	
	\noindent The choice of the weight $\nicefrac{1}{n^\ell}$ in the definition of the rule ensures that the rule will always select a project proposed by the agents with the smallest number of selected projects, who still have unselected projects.
	
	The $k$-equal representation shortlisting rule can be seen as a Thiele rule where the $j$-th weight is defined as $w_j = \sum_{\ell = 0}^{j} \nicefrac{1}{n^\ell}$. Note that in particular this implies that the weight is dependent on the number of agents. 
	
	\medskip
	
	Let us explore the computational complexity of the rule now.
	
	\begin{proposition}
		Let $k \in \mathbb{N}_{> 0}$. There is no algorithm running in polynomial-time that computes, given a shortlisting instance $I$ and profile $\shortProfile$, the outcome of the $k$-equal representation shortlisting rule, unless $\complexP = \complexNP$.
	\end{proposition}
	
	\begin{proof}
		Note that for $k' \in \mathbb{N}_{> 0}$, if all projects have cost $\nicefrac{b}{k'}$, computing the outcome of the $k$-equal representation shortlisting rule amounts to finding a committee of size $k'$ with a Thiele rule with weights $\left(1, \nicefrac{1}{n}, \nicefrac{1}{n^2}, \ldots\right)$ for a multi-winner election \citep{Jans16}. Interestingly, the reduction  presented by \citet{AGGMMW14} to show that proportional approval voting is \complexNP-complete works for all Thiele rules with decreasing weights. Since this is the case here, their reduction also applies.
	\end{proof}
	
	\subsection{Median-Based Shortlisting Rules}
	
	One criterion frequently used for shortlisting in real PB processes is the similarity between the proposals. Since only few projects will be shortlisted, it would be particularly inefficient to shortlist two very similar ones. In the following we rationalise this decision process by introducing a shortlisting rule that clusters the projects and selects representative projects for each cluster.
	
	\medskip
	
	We assume that projects are embedded in a metric space, the distance between two projects being given. Using this metric space, we will try and cluster the proposals submitted during the shortlisting stage.
	
	Formally speaking, we call \emph{distance} any metric over $\conceivProjSet$. For a distance $\delta$, let $\mathit{med}(P)$ be the the \emph{geometric median} of $P \subseteq \conceivProjSet$ defined by:
	\[\mathit{med}(P) = \canonTieRule \left(\left\{\{p\} \mid p \in \argmin_{p^\star \in P} \sum_{p' \in P} \delta(p^\star, p')\right\}\right).\]
	
	A partition of $P$, denoted by $V = \{V_1, \ldots, V_p\}$, is a \emph{$(k, \ell)$-Voronoï partition} with respect to the distance $\delta$, if the representatives of $V$ cost no more than $k \cdot b$ in total:
	\[\sum_{V_j \in V} c(\mathit{med}(V_j)) \leq k \cdot b,\]
	and for every distinct $V_j, V_{j'} \in V$ and every project $p \in V_j$, we have:
	\begin{itemize}
		\item $\delta(p, \mathit{med}(V_j)) \leq \delta(p, \mathit{med}(V_{j'}))$, \textit{i.e.}, every project is in the cluster of its closest geometric median; and
		\item $\delta(p, \mathit{med}(V_j)) \leq \ell$, \textit{i.e.}, $p$ is within distance $\ell$ of $\mathit{med}(V_j)$.
	\end{itemize}
	Thus, the parameter $k$ bounds the total cost of the representatives, and the parameter $\ell$ bounds the maximum distance within a cluster. We denote by $\mathcal{V}_{\delta, k,\ell}(P)$ be the set of all $(k, \ell)$-Voronoï partitions of $P$ with respect to distance $\delta$ and parameters $k$ and $\ell$.
	
	With all these definitions in mind, we are ready to define the class of $k$-median shortlisting rules. As before, $k$ parametrised the maximum cost of the shortlist.
	
	\begin{definition}[$k$-Median Shortlisting Rules]
		\label{def:medianShortRule}
		Let $k \in \mathbb{N}$. The \emph{$k$-median shortlisting rule} $\medianShortRule_{k, \delta}$ with respect to the distance $\delta$ is such that for every shortlisting instance $I$ and profile $\shortProfile$, we have:
		\[\medianShortRule_{k, \delta}(I, \shortProfile) = \canonTieRule\left(\left\{\bigcup_{V_j \in V} \mathit{med}(V_j) \mid V \in \mathcal{V}_{\delta, k, \ell^\star}\left(\bigcup \shortProfile\right) \right\}\right)\]
		where $\ell^\star$ is the smallest $\ell$ such that $\mathcal{V}_{\delta, k,\ell} (\bigcup \shortProfile) \neq \emptyset$.
	\end{definition}
	
	\noindent Note that we chose to minimise $\ell$ in our definition. One can similarly try to minimise $k$, or both $\ell$ and $k$, instead.
	
	It is finally worth saying a few words about the computation complexity of these shortlisting rules. It is straightforward to show that, unless $\complexP = \complexNP$, there cannot be an algorithm that runs in polynomial time and that computes the outcome of a $k$-median shortlisting rule, for any value of $k$ and suitable distance $\delta$. Indeed, when $\delta$ is the Euclidean distance over $\mathbb{R}^2$, our formulation coincide with the $k$-median problem, known to be \complexNP-hard \citep{KaHa79}. Several other results have been published, including approximation algorithms \citep{KMN+04} and fixed-parameters analyses \citep{CGK+19}, and can be used to cope with intractability. 
	
	\section{End-to-End Example}
	\label{sec:initialEx}
	
	We have now introduced all the components of our model. Before getting to the more technical analysis, let us give an example to clarify the whole setting.
	
	\medskip
	
	Consider the following shortlisting instance $I = \tuple{\conceivProjSet, c, b}$ with nine projects, $\conceivProjSet = \{p_1, \ldots p_{9}\}$. Suppose for simplicity that for every project $p \in \conceivProjSet$ we have $c(p) = 1$, \textit{i.e.}, we are in the unit-cost setting. The budget limit is $b = 3$. Consider five agents as described below.
	\begin{center}
		\begin{tabular}{lccc}
			& \begin{tabular}{@{}c@{}}\textbf{Preferences over} \\ \textbf{the Projects}\end{tabular} & \textbf{Awareness set $C_i$} & \begin{tabular}{@{}c@{}}\textbf{Ideal Set} \\ \textbf{Based on $C_i$}\end{tabular} \\
			\midrule
			\textbf{Agent 1} & $p_4 \rhd p_5 \rhd p_1 \rhd p_2 \rhd \cdots$ & $\{p_1, p_2, p_4, p_5\}$ & $\{p_1, p_4, p_5\}$ \\
			\textbf{Agent 2} & $p_1 \rhd p_2 \rhd p_6 \rhd p_4 \rhd \cdots$ & $\{p_2, p_6\}$ & $\{p_2, p_6\}$ \\
			\textbf{Agent 3} & $p_1 \rhd p_2 \rhd p_7 \rhd p_4 \rhd \cdots$ & $\{p_2, p_7\}$ & $\{p_2, p_7\}$ \\
			\textbf{Agent 4} & $p_1 \rhd p_3 \rhd p_8 \rhd p_5 \rhd \cdots$ & $\{p_3, p_8\}$ & $\{p_3, p_8\}$ \\
			\textbf{Agent 5} & $p_1 \rhd p_3 \rhd p_9 \rhd p_5 \rhd \cdots$ & $\{p_3, p_9\}$ & $\{p_3, p_9\}$
		\end{tabular}
	\end{center}
	Assuming agents are truthful they will propose projects according to their ideal sets, computed given their awareness sets. The truthful shortlisting profile would then be:
	\[\shortProfile = (\{p_1, p_4, p_5\}, \{p_2, p_6\}, \{p_2, p_7\}, \{p_3, p_8\}, \{p_3, p_9\}).\]
	Thus, if the nomination shortlisting rule $\nominationShortRule$ is used, the shortlist would be $\conceivProjSet$. In case the $1$-equal representation rule $\reprShortRule_1$ is used, it would be $\{p_1, p_2, p_3\}$.
	
	Suppose the set of shortlisted projects is $\projSet = \conceivProjSet$. Agents are now aware of all the shortlisted projects. They recompute their ideal sets given the new information. Still assuming that agents behave truthfully, the profile for the allocation stage is:
	\[\profile = (\{p_1, p_4, p_5\}, \{p_1, p_2, p_6\}, \{p_1, p_2, p_7\}, \{p_1, p_3, p_8\}, \{p_1, p_3, p_9\}).\]
	This corresponds to the vector of the ideal sets computed by each agent with respect to $\projSet$, and their respective preferences over the projects. With such a profile, if the allocation rule \greWelCost{} is used, the final budget allocation would be $\pi = \{p_1, p_2, p_3\}$.
	
	\section{Axioms for Shortlisting Rules}
	\label{sec:AxiomsShortRule}
	
	We now assess the axiomatic merits of the shortlisting rules we have introduced. 
	
	\medskip
	
	The first axiom we define is \emph{non-wastefulness}. It requires that no amount of the budget should be wasted because not enough projects were shortlisted.
	
	\begin{definition}[Non-Wastefulness]
		\label{def:nonWastefulness}
		A shortlisting rule $\shortRule$ is \emph{non-wasteful} if for every shortlisting instance $I = \tuple{\conceivProjSet, c, b}$ and profile $\shortProfile$, one of the following two holds:
		\[c(\shortRule(I, \shortProfile)) \geq b \qquad \text{ or } \qquad \shortRule(I, \shortProfile) = \bigcup \shortProfile.\]
	\end{definition}
	
	\noindent This axioms can be interpreted as an efficiency requirement ensuring that no money is wasted because of the shortlisting rule.
	
	\medskip
	
	We believe that another important property of a shortlisting rule is that every agent is represented in the outcome. This is particularly relevant in the shortlisting stage since any subset of $\conceivProjSet$ is theoretically admissible.
	
	\begin{definition}[Representation Efficiency]
		For a given shortlisting instance $I = \tuple{\conceivProjSet, c, b}$ and a given shortlisting profile $\shortProfile$, a set of projects $P \subseteq \conceivProjSet$ is \emph{representatively dominated} if there is a set $P' \subseteq \conceivProjSet$ with $c(P') \leq c(P)$, and $|P' \cap P_i| \geq |P \cap P_i|$ for all $i \in \agentSet$, with a strict inequality for at least one agent.
		
		A shortlisting rule $\shortRule$ is \emph{representatively efficient} if for every shortlisting instance $I = \tuple{\conceivProjSet, c, b}$, and every shortlisting profile $\shortProfile$, $\shortRule(I, \shortProfile)$ is not representatively dominated by any other subset of projects.
	\end{definition}
	
	\noindent A set of projects $P$ is thus representatively dominated by another one $P'$ if $P'$ does not cost more than $P$, and, for every agent, at least as many projects that they submitted have been selected in $P'$ as in $P$, and strictly more for at least one of them.
	
	This axiom provides guarantees that the shortlisting rule is aiming to achieve some kind of representation. Note however that the guarantee is not very strong and can lead to large disparities between the agents: some could have all their proposals shortlisted, and some others none, in a shortlist that is still representatively efficient.
	
	\medskip
	
	These are the two axioms with respect to which we will analyse the shortlisting rules. This analysis is presented below.
	
	\medskip
	
	We will start with the nomination shortlisting rule, that trivially satisfies both non-wastefulness and representation efficiency.
	
	\begin{proposition}
		\label{prop:Axioms_Nom_Short_Rule}
		The nomination shortlisting rule is non-wasteful and representatively efficient.
	\end{proposition}
	
	\begin{proof}
		With the nomination shortlisting rule $\nominationShortRule$, for every shortlisting instance $I$ and profile $\shortProfile$, we have $\nominationShortRule(I, \shortProfile) = \bigcup \shortProfile$. Thus, the second condition of non-wastefulness is always trivially satisfied.
		
		Given that every project is shortlisted, the shortlist $\nominationShortRule(I, \shortProfile)$ cannot be representatively dominated. $\nominationShortRule$ is thus representatively efficient.
	\end{proof}
	
	We now move to the $k$-equal representation shortlisting rule, showing that it is both non-wasteful and representatively efficient, as long as $k$ is at least 2.
	
	\begin{proposition}
		\label{prop:Axioms_Repr_Short_Rule}
		For every $k \geq 2$, the $k$-equal representation shortlisting rule is non-wasteful, but it is not for $k = 1$. Moreover, for every $k \geq 1$, the $k$-equal representation shortlisting rule is representatively efficient.
	\end{proposition}
	
	\begin{proof}
		Let us first prove that for every $k \geq 2$, the $k$-equal representation shortlisting rule $\shortRule$ is non-wasteful. Suppose it is not, then, there would exist a shortlisting instance $I = \tuple{\shortProfile, c, b}$ and a shortlisting profile $\shortProfile$ such that:
		\[c(\reprShortRule_k(I, \shortProfile)) < b \qquad \text{ and } \qquad \reprShortRule_k(I, \shortProfile) \neq \bigcup \shortProfile.\]
		From this, we know that there exists a project $p \in \bigcup \shortProfile$ that has not been shortlisted, \textit{i.e.}, such that $p \notin \reprShortRule_k(I, \shortProfile)$. Thus, the representation score of the set $\reprShortRule_k(I, \shortProfile) \cup \{p\}$ is higher than that of $\shortRule(I, \shortProfile)$. 
		Moreover, for any $k \geq 2$, the facts that $c(p) \leq b$ and $c(\reprShortRule_k(I, \shortProfile)) < b$ together imply that:
		\[c(\shortRule(I, \shortProfile) \cup \{p\}) \leq 2 \cdot b \leq k \cdot b.\]
		Overall, if such a project $p$ exists, then $\reprShortRule_k(I, \shortProfile) \cup \{p\}$ is an admissible outcome of $\reprShortRule_k$ with a higher total weight than $\reprShortRule_k(I, \shortProfile)$. This contradicts the definition of $\reprShortRule_k$, and thus proves that it is non-wasteful.
		
		Note that for $k = 1$, the definition of $\reprShortRule_k$ implies that the cost of the shortlist will not be more than $b$ (Definition \ref{def:equalReprShortRule}), and non-wastefulness requires the same cost to be at least $b$ (Definition \ref{def:nonWastefulness}). Since projects are indivisible, it is clearly not always possible to shortlist a set of projects of cost exactly $b$.
		
		\medskip
		
		We now show that for every $k \geq 1$, the $k$-equal representation shortlisting rule is representatively efficient. The proof is actually trivial, it is immediately derived from the choice of the weight $\nicefrac{1}{n}$ in the definition of the rule. Indeed, since $\nicefrac{1}{n} > 0$, a representatively dominated shortlist would always have a lower total score and would thus not be selected.
	\end{proof}
	
	Now comes the turn of the median shortlisting rules. We prove that these rules are non-wasteful, but not representatively efficient.
	
	\begin{proposition}
		\label{prop:Axioms_Median_Short_Rule}
		Let $\delta$ be an arbitrary distance over $\conceivProjSet$. The following facts hold:
		\begin{itemize}
			\item For every $k \geq 2$, the shortlisting rule $\medianShortRule_{k, \delta}$ is non-wasteful;
			\item There exists no $k \in \mathbb{N}_{> 0}$ such that $\medianShortRule_{k, \delta}$ is representatively efficient.
		\end{itemize}
	\end{proposition}
	
	\begin{proof}
		The proof that for $k \geq 2$, the $\medianShortRule_{k, \delta}$ is non-wasteful is similar to that of Proposition~\ref{prop:Axioms_Repr_Short_Rule}. To see why, note that for every shortlisting instance $I$ and profile $\shortProfile$, there will never be an unselected project $p \in \bigcup \shortProfile \setminus \shortRule(I, \shortProfile)$ such that $c(\shortRule(I, \shortProfile) \cup \{p\}) \leq k \cdot b$. Indeed, if such a $p$ exists, selecting it would always lead to a smaller within-cluster distance, simply by including $p$ as its own cluster (since by the definition of a metric, the distance between $p$ and any other project $p' \in \conceivProjSet \setminus \{p\}$ is non-zero). We can thus reach the same contradiction that we reached in the proof of Proposition~\ref{prop:Axioms_Repr_Short_Rule}.
		
		\medskip
		
		It is also easy to see that $\medianShortRule_{k, \delta}$ is not efficiently representative. We do not provide a formal proof here as the correctness of the statement should be intuitively clear. It is derived from the fact that for any shortlisting profiles $\shortProfile$ and $\shortProfile'$, such that $\bigcup \shortProfile = \bigcup \shortProfile'$, the outcome of $\medianShortRule_{k, \delta}$ would be the same. This means that $\medianShortRule_{k, \delta}$ is completely oblivious of the identity of the agents, making it fail representation efficiency.
	\end{proof}
	
	The axiomatic analysis of the shortlisting rules is now complete. We will move on to the next focus point of this paper: the interactions between the two stages.
	
	\section{First-Stage Strategy-Proofness}
	\label{sec:FSSP}
	
	We now turn to the analysis of strategic interactions during the shortlisting stage. Remember our motivational example in the introduction, we wondered whether Sophie should propose her project about the fountain or not, because of its impact on the final decision, taken during the second stage. We have hinted at reasons why it would actually be better for her not to.
	Throughout this section, we will study such strategic behaviour and investigate whether it can be prevented or not.
	
	\medskip
	
	One of the main challenges to formalise the concept of strategic behaviour during the first stage, is that agents actually reason about the outcome of the process---the final budget allocation---that is only decided one stage later. Let us then take the time to discuss the information available to an agent willing to strategise, the \emph{manipulator}.
	
	In the classical voting framework \citep{Zwic16}, it is assumed that the potential manipulator has access to all the other ballots before submitting their own.
	In our setting, when considering a manipulator choosing which proposal to submit during the first stage, the same assumption is reasonable with respect to the proposals submitted by the other agents during the first stage, but not with respect to the ballots the other agents are going to submit during the second stage, only \emph{after} the shortlist will have been determined.
	Indeed, the set of actions for the second stage depends on the proposal of the manipulator in the first stage.
	We thus need to reason about the outcome of the second stage given the profile that the manipulator \emph{expects} to occur.
	
	We explore three possibilities. In the first two cases, a manipulator in the first stage is unsure what will happen during the second stage, but assumes that either the worst scenario will be realised (\emph{pessimistic manipulation}) or the best one (\emph{optimistic manipulation}). In the third case, they know the other agents' true preferences and trusts they will vote accordingly (\emph{anticipatory manipulation}).
	
	Because there are no reasons to assume that a potential manipulator would only behave strategically in the first stage, and not in the second stage, we also need the concept of a best response in the second stage. For that, we introduce some further notation. For a given allocation rule $\pbRule$, allocation instance $I = \tuple{\projSet, c, b}$, profile $\profile$, and agent $i \in \agentSet$, let $A_i^\star(I, \profile, \pbRule)$ be the best response of $i$ to $\profile$, defined such that:
	\[A_i^\star(I, \profile, \pbRule) = \canonTieRule(\{A_i' \in \projSet \mid \pbRule(I, (\profile_{-i}, A_i')) = \canonTieRule(\mathfrak{P}^\star)\}),\]
	where $\mathfrak{P}^\star = \undom(\succsim_{\idealSet_i(\projSet)}, \{\pbRule(I, (\profile_{-i}, P)) \mid P \subseteq \projSet\})$ and $\succsim$ is generated given an arbitrary completion principle.
	
	Let us unravel a bit this definition. $\succsim_{\idealSet_i(\projSet)}$ is the weak order over subsets of projects that is induced by the completion principle in use, based on $\idealSet_i(\projSet)$, the ideal set of agent $i$. Then, $\mathfrak{P}^\star$ is the set of undominated budget allocations returned by the allocation rule $\pbRule$, for any approval ballot $P$, agent $i$ can submit in the second stage, where domination is defined with respect to $\succsim_{\idealSet_i(\projSet)}$. Because we need a single outcome, we break ties between the, potentially, several such undominated budget allocation. Then, $A_i^\star(I, \profile, \pbRule)$ is the ballot that achieved this aforementioned budget allocation and that is selected by the tie-breaking rule (since the relevant may be reachable via several ballots). The intuition is that $A_i^\star(I, \profile, \pbRule)$ is the best ballot agent $i$ can submit in the second stage given $I$, $\profile$, and $F$, and our assumptions on the preferences of the agents.
	When clear from the context, we omit $I$, $\profile$, and/or $\pbRule$ from the notation $A_i^\star(I, \profile, \pbRule)$.
	
	We are now ready to properly formalise all we described above. This is the aim of the following definition.
	
	\begin{definition}[Successful Manipulation]
		Let $\shortRule$ be a shortlisting rule, 
		$\pbRule$ an allocation rule, 
		$I_1 = \tuple{\conceivProjSet, c, b}$ a shortlisting instance, 
		$\shortProfile$ a shortlisting profile,
		and $P_i' \subseteq \conceivProjSet$ an alternative proposal for agent $i\in\agentSet$.
		Consider the shortlists $\projSet = \shortRule(I_1, \shortProfile)$ and $\projSet' = \shortRule(I_1, (\shortProfile_{-i}, P_i'))$, determining 
		the allocation instances $I_2 = \tuple{\projSet, c, b}$ and $I'_2 = \tuple{\projSet', c, b}$.
		
		For any two profiles, $\profile$ for $\projSet$ and $\profile'$ for $\projSet'$, we simplify the notation by defining the two following abbreviations:
		\begin{align*}
			F^\star(I_2, \profile) & = F(I_2, (\profile_{-i},A_i^\star(I_2,\profile))), \\
			F^\star(I'_2, \profile') & = F(I'_2, (\profile'_{-i},A_i^\star(I'_2,\profile'))).
		\end{align*}
		To clarify, $F^\star(I_2, \profile)$ is thus the final budget allocation for the instance $I_2$ and profile $\profile$ in which agent $i$ is playing their best response. The case of $F^\star(I'_2, \profile')$ is analogous, for the instance $I_2'$ and profile $\profile'$.
		
		Then, for a given completion principle generating $\succsim_{\idealSet_i(\projSet\cup\projSet')}$, we say that:
		\begin{itemize}
			
			\item $P'_i$ is a \emph{successful pessimistic manipulation} if, 
			for all profiles $\profile$ on $\projSet$ and $\profile'$ on $\projSet'$,
			it is the case that $F^\star(I_2', \profile') \succsim_{\idealSet_i(\projSet\cup\projSet')} F^\star(I_2, \profile)$, 
			with a strict preference for at least one pair~$(\profile,\profile')$.
			
			\item $P'_i$ is a \emph{successful optimistic manipulation} if,
			for at least one profile $\profile$ on $\projSet$ and one profile $\profile'$ on $\projSet'$,
			it is the case that $F^\star(I_2', \profile') \succ_{\idealSet_i(\projSet\cup\projSet')} F^\star(I_2, \profile)$.
			
			\item $P'_i$ is a \emph{successful anticipatory manipulation} if,
			for the two profiles~$\profile = \topProfile(\projSet)$ and $\profile' = \topProfile(\projSet')$,
			it is the case that $F^\star(I_2', \profile') \succ_{\idealSet_i(\projSet\cup\projSet')} F^\star(I_2, \profile)$.
			
		\end{itemize}
	\end{definition}
	
	\noindent 
	Thus, a pessimist is pessimistic with respect to the advantages they can gain from manipulating: assuming the best if she is truthful and the worst otherwise. 
	For optimists it is the other way around. Finally, an anticipatory manipulator knows everyone's preferences on both $\projSet$ and $\projSet'$ and uses them to predict their votes for the second stage.
	
	We are looking for rules that do not allow for successful manipulation, \textit{i.e.}, that are \emph{first-stage strategy-proof} (FSSP). We distinguish two cases: either the manipulator is \emph{restricted to} their awareness set (R-FSSP) or they can also propose any of the projects proposed by others (during the shortlisting stage), \textit{i.e.}, they are \emph{unrestricted} (U-FSSP).
	\begin{figure}
		\centering
		\begin{tikzpicture}[shorten < = 3pt, shorten > = 3pt]
			\node[align = center] (info) at (0, 0) {What information is available\\ to the manipulator $i \in \agentSet$?};
			
			\node[align = center] (R-FSSP) at (-4, -2) {Only their own\\awareness set $C_i$\\\emph{R-FSSP}};
			\node[align = center] (U-FSSP) at (4, -2) {Their awareness set $C_i$ and \\the proposals of the other agents\\\emph{U-FSSP}};
			\path[->] (info) edge (R-FSSP);
			\path[->] (info) edge (U-FSSP);
			
			\node[align = center] (anticipation) at (0, -4) {What is the manipulator\\anticipating for the allocation stage?};
			\path[->] (R-FSSP) edge (anticipation);
			\path[->] (U-FSSP) edge (anticipation);
			
			\node[] (worst) at (0, -6) {The worst};
			\node[] (best) at (0, -7) {The best};
			\node[align = center] (truth) at (0, -8) {The others to\\behave truthfully};
			\path[->, shorten < = 1pt, shorten > = 1pt] (anticipation) edge (worst);
			
			\node[] (R-FSSP-P) at (-4, -6) {\emph{R-FSSP-P}};
			\node[] (R-FSSP-O) at (-4, -7) {\emph{R-FSSP-O}};
			\node[] (R-FSSP-A) at (-4, -8) {\emph{R-FSSP-A}};
			\path[->, dashed, my-gray-light] (R-FSSP) edge (R-FSSP-P);
			\path[->] (worst) edge (R-FSSP-P);
			\path[->] (best) edge (R-FSSP-O);
			\path[->] (truth) edge (R-FSSP-A);
			
			\node[] (U-FSSP-P) at (4, -6) {\emph{U-FSSP-P}};
			\node[] (U-FSSP-O) at (4, -7) {\emph{U-FSSP-O}};
			\node[] (U-FSSP-A) at (4, -8) {\emph{U-FSSP-A}};
			\path[->, dashed, my-gray-light] (U-FSSP) edge (U-FSSP-P);
			\path[->] (worst) edge (U-FSSP-P);
			\path[->] (best) edge (U-FSSP-O);
			\path[->] (truth) edge (U-FSSP-A);
		\end{tikzpicture}
		\caption{Explanation of first-stage strategy-proofness concepts.}
		\label{fig:taxonomyFSSP}
	\end{figure}
	
	\begin{definition}[First-Stage Strategy-Proofness]
		\index{First-stage strategy-proofness}
		For a given completion principle, a pair $\tuple{\shortRule, \pbRule}$ consisting of a shortlisting rule $\shortRule$ and an allocation rule $\pbRule$ is said to be \emph{restricted-FSSP} (R-FSSP) with respect to a given type of manipulation if for every shortlisting instance $\tuple{\conceivProjSet, c, b}$, every awareness profile $\awareStruct = (C_1, \ldots, C_n)$, every shortlisting profile $\shortProfile = (P_1, \ldots, P_n)$ where $P_{i'} \subseteq C_{i'}$ for all $i' \in \agentSet$, and every agent $i \in \agentSet$, there is no $P_i' \subseteq C_i$ such that submitting $P_i'$ instead of $\idealSet_i(C_i)$ is a successful manipulation for $i$.
		
		In case $P_i' \subseteq C_i \cup \bigcup \shortProfile$ and we consider $\idealSet_i(C_i \cup \bigcup \shortProfile)$ instead of $\idealSet_i(C_i)$ in the above, we say that $\tuple{\shortRule,\pbRule}$ is \emph{unrestricted-FSSP} (U-FSSP).
	\end{definition}
	
	\noindent 
	Thus, in the unrestricted case, agents are assumed to \emph{first} gain access to everyone's proposals and \emph{then} decide whether or not to vote truthfully.
	
	We introduce some further abbreviations. Let FSSP-P stand for FSSP with respect to pessimistic manipulation attempts, FSSP-O for FSSP with respect to optimistic manipulation attempts, and FSSP-A for FSSP with respect to anticipatory manipulation attempts.
	We have thus introduced six different FSSP concepts in total. A simplified overview is given in Figure~\ref{fig:taxonomyFSSP} to clarify everything.
	
	\medskip
	
	It should be clear, at least from the text around the definitions that there are some links between the different FSSP concepts we have introduced.
	The following result summarises how the different notions introduced relate to each other, where $\mathfrak{X}$ implying $\mathfrak{X}'$ means that any pair $\tuple{\shortRule, \pbRule}$ satisfying $\mathfrak{X}$ also satisfies $\mathfrak{X}'$.
	
	\begin{proposition}
		\label{prop:FSSPimplications}
		The following implications hold for any given completion principle:
		\begin{itemize}
			\item R-FSSP-O implies R-FSSP-A and R-FSSP-P.
			\item U-FSSP-O implies U-FSSP-A and U-FSSP-P.
			\item R-FSSP implies U-FSSP for all types of manipulation.
		\end{itemize}
	\end{proposition}
	
	\begin{proof}
		The first two claims are immediately derived from the relevant definitions.
		To see that the last of these claims is also true, observe that U-FSSP is a special case of R-FSSP, namely when the manipulator can conceive of all the proposed projects, \textit{i.e.}, when $C_i = \bigcup \shortProfile$.
	\end{proof}
	
	\noindent Interestingly, the link between pessimistic and anticipatory manipulations is not clear. Although if a successful pessimistic manipulation exists it ensures that an anticipative manipulation would lead to a weakly better outcome, nothing guarantees that this outcome would be strictly better for the manipulator.
	
	\subsection{Awareness-Restricted Manipulation}
	\label{sec:RestrictedFSSP}
	
	We start by proving an impossibility theorem stating that no pair of reasonable rules can be first-stage strategy-proof when manipulators are restricted to their awareness sets.
	By ``reasonable rule'' we mean a non-wasteful shortlisting rule, followed by a \emph{determined} allocation rule.
	
	\begin{definition}[Determined]
		\index{Determined rule}
		An allocation rule $\pbRule$ is \emph{determined} if, for every allocation instance $I = \tuple{\projSet, c, b}$, and every profile $\profile$, we have $\pbRule(I, \profile) \neq \emptyset$.
	\end{definition}
	
	\begin{theorem}
		\label{thm:Imposibility_R_FSSP}
		Every pair $\tuple{\shortRule, \pbRule}$ of a non-wasteful shortlisting rule $\shortRule$ and a determined allocation rule $\pbRule$ is neither R-FSSP-P nor R-FSSP-A (and thus also not R-FSSP-O), for any completion principle that is top-first.
	\end{theorem}
	
	\begin{proof}
		We provide a proof for R-FSSP-P, but the same proof also goes through for R-FSSP-A. The claim for R-FSSP-O then follows from Proposition~\ref{prop:FSSPimplications}.
		
		Let $I = \tuple{\conceivProjSet, c, b}$ be the shortlisting instance with two conceivable projects $p_1$ and $p_2$, both of cost 1, and a budget limit $b = 1$. Suppose there are two agents. The preferences of the first agent are such that $p_2 \rhd_1 p_1$, their awareness set is $C_1 = \{p_1\}$. For the second agent, we have $p_1 \rhd_2 p_2$, and $C_2 = \{p_2\}$. Overall, each agent is aware only of the project they like less. The truthful shortlisting profile is then $\shortProfile = (\{p_1\}, \{p_2\})$.
		
		Assuming that $\shortRule$ is non-wasteful, we know that $|\shortRule(I, \shortProfile)| \geq 1$. There are thus three possible cases for $\shortRule(I, \shortProfile)$: to shortlist either just $p_1$, just $p_2$, or both $p_1$ and $p_2$. Let us go through each of them independently.
		
		In case $\shortRule(I, \shortProfile) = \{p_1\}$, whichever way the agents vote in the allocation stage, as $\pbRule$ is assumed to be determined, the final budget allocation must be $\{p_1\}$.
		Now, if agent~1 manipulates by not proposing any project during the first stage, only project $\{p_2\}$ will get shortlisted (and this has to happen since $\shortRule$ is non-wasteful). In that case, $\{p_2\}$ will also be the final budget allocation, given that $\pbRule$ is determined. Since $\{p_2\}$ is the ideal point of agent~1 for the set of projects $\{p_1, p_2\}$, they would strictly prefer $\{p_2\}$ over $\{p_1\}$ for any completion principle that is top-first. So agent~1 has an incentive to pessimistically manipulate.
		
		The case of $\shortRule(I, \shortProfile) = \{p_2\}$ is perfectly analogous to the previous one: the final budget allocation under truthful behaviour would be $\{p_2\}$, but then agent~2 has an incentive to pessimistically manipulate by not submitting $p_2$ during the first stage so that the final outcome would be $\{p_1\}$.
		
		Finally, consider the case $\shortRule(I, \shortProfile) = \{p_1, p_2\}$. Suppose the final budget allocation is $\{p_1\}$ in case both agents vote truthfully. Then, just as in the first case, agent~1 has an incentive to submit an empty set of proposals instead, as that guarantees a final budget allocation of $\{p_2\}$. In the analogous case where the final budget allocation is $\{p_2\}$, agent~2 would pessimistically manipulate.
		
		Overall, there always is an agent who has an incentive to pessimistically manipulate. $\tuple{\shortRule, \pbRule}$ is thus not R-FSSP-P.
	\end{proof}
	
	\noindent
	Note that the scenario used in the proof shows that unrestricted-FSSP does \emph{not} imply restricted-FSSP. Indeed, under U-FSSP, no agent would have an incentive to manipulate in this scenario, as they would have all the information they need to submit an optimal truthful proposal.
	
	Regarding the specific shortlisting rules we have introduced, we can now derive the following corollary.
	
	\begin{corollary}
		Let $\pbRule$ be an allocation rule that is exhaustive, $k \geq 2$, and $\delta$ an arbitrary distance over $\conceivProjSet$. Then, none of the pairs $\tuple{\reprShortRule_k, \pbRule}$,  $\tuple{\medianShortRule_{k, \delta}}$ or $\tuple{\nominationShortRule, \pbRule}$ are R-FSSP-P, R-FSSP-A, or R-FSSP-O, for any completion principle that is top-first.
	\end{corollary}
	
	\begin{proof}
		The proof is immediately derived from Theorem~\ref{thm:Imposibility_R_FSSP} and the fact that the relevant shortlisting rules are non-wasteful (Propositions~\ref{prop:Axioms_Nom_Short_Rule}, \ref{prop:Axioms_Repr_Short_Rule} and~\ref{prop:Axioms_Median_Short_Rule}).
	\end{proof}
	
	\subsection{Unrestricted Manipulation}
	
	We now turn to the case where the manipulator gains awareness by looking at the projects already submitted for the first stage.
	
	\medskip
	
	Let us start with the nomination rule. We will show that it is immune to pessimistic manipulation when used with allocation rules that are \emph{unanimous}, a new axiom we introduce below.
	
	\begin{definition}[Unanimity]
		\index{Unanimity}
		An allocation rule $\pbRule$ is \emph{unanimous} if, for every allocation instance $I = \tuple{\projSet, c, b}$ and every feasible subset of projects $A \in \allocSet(I)$, it is the case that for the profile $\profile = (A, \ldots, A)$, we have:
		\[\pbRule(I, \profile) \supseteq A.\]
	\end{definition}
	
	\noindent This axioms states that if every agent submits the same feasible ballot, then the set of projects in this ballot should be part of the outcome. This is a rather weak axiom and every allocation rule we have defined satisfies it.
	
	Let us now state our result for the nomination shortlisting rule.
	
	\begin{proposition}
		\label{prop:NomShortRule_UFSSPP}
		The pair $\tuple{\nominationShortRule, \pbRule}$ where $\pbRule$ is an allocation rule that is unanimous is U-FSSP-P for every completion principle that is top-first.
	\end{proposition}
	
	\begin{proof}
		Let $I = \tuple{\conceivProjSet, c, b}$ be a shortlisting instance, $\awareStruct$ an awareness profile, 
		and $\shortProfile$ a shortlisting profile. Consider an agent $i^\star \in \agentSet$ and let $P_{i^\star} = \idealSet_{i^\star}(C_{i^\star} \cup \bigcup \shortProfile_{-i^\star})$. Denote $\projSet = \nominationShortRule(I, (\shortProfile_{-i^\star}, P_{i^\star}))$, and observe that $P_{i^\star} \subseteq \projSet$ because of the definition of $\nominationShortRule$. Moreover, from the definition of \nominationShortRule{}, we know that if agent ${i^\star}$ submits $P_{i^\star}'$ instead of $P_{i^\star}$, with $P_{i^\star}' \neq P_{i^\star}$ the shortlist will become $\projSet' = P_{i^\star}' \cup \left(\bigcup_{i \in \agentSet \setminus \{{i^\star}\}}P_{i}\right)$. This implies that:
		\begin{align}
			\projSet \cap P_{i^\star} \supseteq \projSet' \cap P_{i^\star}, \label{eq:Nom_Short_Rule_U_FSSP_Line1}
		\end{align}
		keeping in mind that $P_{i^\star}$ is the ideal set of $i^\star$ for $C_{i^\star} \cup \bigcup \shortProfile$. Thus, under $(\shortProfile_{-i^\star}, P_{i^\star}')$, it cannot be that more projects from $\idealSet_{i^\star}(C_{i^\star} \cup \bigcup \shortProfile)$ are shortlisted than under $\shortProfile$. 
		
		We focus on the second stage now. Consider first the profile $\profile$ in which all agents submit $\idealSet_{i^\star}(\projSet)$. Clearly, submitting $\idealSet_{i^\star}(\projSet)$ is a best response for $i^\star$ in this case, so $A_{i^\star}^\star(\tuple{\projSet, c, b}, \profile)) = \idealSet_{i^\star}(\projSet)$. Since $\pbRule$ is unanimous, we thus have $\pbRule(\tuple{\projSet, c, b}, \profile) = \idealSet_{i^\star}(\projSet)$. Now, since $\projSet \cup \projSet' \subseteq C_{i^\star} \cup \bigcup \shortProfile$ and $P_{i^\star} \subseteq \projSet$, we know that:
		\[\idealSet_{i^\star}(\projSet \cup \projSet') =  P_{i^\star} = \idealSet_{i^\star}(\projSet).\]
		Given that we assumed the completion principle to be top-first, it is clear that no budget allocation $\pi \in \allocSet(\tuple{\projSet', c, b})$ will be strictly preferred to $\pbRule(\tuple{\projSet, c, b}, \profile)$ by $i^\star$. This directly implies that submitting $P_{i^\star}'$ cannot be a successful pessimistic manipulation for $i^\star$.
	\end{proof}
	
	Because anticipative manipulation is defined for one very specific profile of the second stage, it is harder to get general results. Still, we can show that any pair consisting of the nomination shortlisting rule $\nominationShortRule$ and one of the allocation rules we have introduced satisfy U-FSSP-A.
	
	\begin{proposition}
		\label{prop:NomShortRule_UFSSPA}
		The pair $\tuple{\nominationShortRule, \pbRule}$ is not U-FSSP-A, and thus not U-FSSP-O, when $\pbRule$ is one of \greWelCard{}, \greWelCost{}, \maxWelCard{}, \maxWelCost{}, \seqPhragmen{}, \maximinSupport{}, \mesSat{$\cardSatisfaction$}, or \mesSat{$\costSatisfaction$}, for every completion principle that satisfies cost-neutral monotonicity.
	\end{proposition}
	
	\begin{proof}
		Recall our initial example presented in Section~\ref{sec:initialEx}. In case agents submit their proposal for the first stage truthfully, the shortlist under $\nominationShortRule$ is $\projSet = \conceivProjSet = \{p_1, \ldots, p_9\}$. In this case, the truthful profile for the allocation stage is:
		\[\profile = (\{p_1, p_4, p_5\}, \{p_1, p_2, p_6\}, \{p_1, p_2, p_7\}, \{p_1, p_3, p_8\}, \{p_1, p_3, p_9\}).\]
		For the allocation instance $I = \tuple{\projSet, c, b}$, the different rules we are considering produce the following outcomes:
		\begin{gather*}
			\begin{align*}
				\greWelCard(I, \profile) & = \greWelCost(I, \profile) = \maxWelCard(I, \profile) = \maxWelCost(I, \profile) \\ & = \seqPhragmen(I, \profile) = \maximinSupport(I, \profile) = \{p_1, p_2, p_3\},
			\end{align*} \\[0.2em]
			\mesSat{\cardSatisfaction}(I, \profile) = \mesSat{\costSatisfaction}(I, \profile) = \{p_1\}.
		\end{gather*}
		
		Suppose now that agent~1 submits $P_1' = \{p_4, p_5\}$ instead of $P_1 = \{p_1, p_4, p_5\}$ in the shortlisting stage. The shortlist computed by $\nominationShortRule$ then becomes $\projSet' = \conceivProjSet \setminus \{p_1\}$. After the agents have recomputed their ideal set, and assuming that they behave truthfully, the profile in the second stage would then be:
		\[\profile' = (\{p_4, p_5\}, \{p_2, p_4, p_6\}, \{p_2, p_4, p_7\}, \{p_3, p_5, p_8\}, \{p_3, p_5, p_9\}).\]
		For the new allocation instance $I' = \tuple{\projSet', c, b}$, the different rules we are considering produce the following outcomes:
		\begin{gather*}
			\begin{align*}
				\greWelCard(I, \profile') & = \greWelCost(I, \profile') = \maxWelCard(I, \profile') \\ & = \maxWelCost(I, \profile') = \seqPhragmen(I, \profile') \\ & = \maximinSupport(I, \profile') = \{p_2, p_4, p_5\},
			\end{align*} \\[0.2em]
			\mesSat{\cardSatisfaction}(I, \profile') = \mesSat{\costSatisfaction}(I, \profile') = \{p_4, p_5\}.
		\end{gather*}
		
		Let us now check that $P_1'$ is a successful anticipative manipulation for agent~1. Their ideal set across the two scenarios is $\idealSet_1(\projSet \cup \projSet') = \{p_1, p_4, p_5\}$. Thus in the first case only $p_1$ is in the intersection of the outcomes of the rules with $\idealSet_1(\projSet \cup \projSet')$. In the second case---when agent~1 submits $P_1'$---this intersection includes both $p_4$ and $p_5$. Given that we assumed the completion principle be cost-neutral monotonic, the anticipative manipulation of agent~1 is thus successful.
		
		The statement for U-FSSP-A then follows from Proposition~\ref{prop:FSSPimplications}.
	\end{proof}
	
	Moving on to other shortlisting rules, we can show that they are not immune to manipulation. We will prove that this is the case when paired with either \emph{unanimous} or \emph{determined} allocation rules.
	
	We first, prove this for the $k$-equal-representation shortlisting rule.
	
	\begin{proposition}
		\label{prop:reprShortRuleFSSP}
		For all $k \in \mathbb{N}_{> 0}$, the pair $\tuple{\reprShortRule_k, \pbRule}$ where $\pbRule$ is a unanimous allocation rule, is neither U-FSSP-P nor U-FSSP-O, for any completion principle that satisfies top-necessity and top-sufficiency.
		
		Moreover, if $\pbRule$ is determined, then the pair $\tuple{\reprShortRule_1, \pbRule}$ is not U-FSSP-A for any completion principle that satisfies top-sufficiency.
	\end{proposition}
	
	\begin{proof}
		We first prove the claim for $k = 1$ and then explain how to generalise to any $k \in \mathbb{N}_{>0}$ (only for U-FSSP-P and U-FSSP-O). Let $I = \tuple{\conceivProjSet, c, b}$ be a shortlisting instance with $\conceivProjSet = \{p_1, \ldots p_4\}$, $c(p_2) = 2$, $c(p) = 1$ for all $p \in \conceivProjSet \setminus \{p_2\}$, and $b = 2$. We consider three agents with the following preferences.
		\begin{center}
			\begin{tabular}{lccc}
				& \begin{tabular}{@{}c@{}}\textbf{Preferences over} \\ \textbf{the Projects}\end{tabular} & \textbf{Awareness set $C_i$} & \begin{tabular}{@{}c@{}}\textbf{Ideal Set} \\ \textbf{Based on $C_i$}\end{tabular} \\
				\midrule
				\textbf{Agent 1} & $p_3 \rhd p_4 \rhd p_2 \rhd p_1$ & $\{p_1, p_2, p_3, p_4\}$ & $\{p_3, p_4\}$ \\
				\textbf{Agent 2} & $p_1 \rhd p_2 \rhd p_3 \rhd p_4$ & $\{p_1, p_2\}$ & $\{p_1, p_3\}$ \\
				\textbf{Agent 3} & $p_2 \rhd p_1 \rhd p_3 \rhd p_4$ & $\{p_2\}$ & $\{p_2\}$
			\end{tabular}
		\end{center}
		
		\noindent Consider the $1$-equal-representation shortlisting rule. Under the truthful profile $\shortProfile = (\{p_3, p_4\}, \{p_1, p_2\}, \{p_2\})$, the shortlist would be $\projSet = \{p_2\}$.
		
		Assume now that agent~1 submits $\{p_1, p_3\}$ instead of $\{p_3, p_4\}$, leading to the shortlisting profile $\shortProfile' = (\{p_3, p_4\})$. Then, because of the tie-breaking mechanism, the outcome of the first stage becomes $\projSet' = \{p_1, p_3\}$.
		
		The ideal set of agent~1 across both scenarios is $\idealSet_1(\projSet \cup \projSet') = \{p_1, p_3\} = \projSet'$ which implies that $\idealSet_1(\projSet \cup \projSet') \cap \projSet = \emptyset$. Thus, we have the following two facts:
		\begin{itemize}
			\item $\pi \cap \idealSet_1(\projSet \cup \projSet') = \emptyset$ for every $\pi \in \allocSet(\tuple{\projSet, c, b})$;
			\item $\pi' \cap \idealSet_1(\projSet \cup \projSet') \neq \emptyset$ for every $\pi' \in \allocSet(\tuple{\projSet', c, b})$ such that $\pi' \neq \emptyset$.
		\end{itemize}
		So every budget allocation in $\allocSet(\tuple{\projSet', c, b})$ that is non-empty will be strictly preferred by agent~1 to all of the ones in $\allocSet(\tuple{\projSet, c, b})$ (since the completion principle is top-sufficient). Moreover, agent~1 is indifferent between the empty budget allocation and any one from $\allocSet(\tuple{\projSet', c, b})$ (since the completion principle is top-necessary). Given that $\pbRule$ is unanimous, any budget allocation from $\allocSet(\tuple{\projSet', c, b})$ can be achieved (by the corresponding unanimous profile). This means that agent~1's manipulation is pessimistically successful.
		
		In case $\pbRule$ is determined, for any profile $\profile$ for the allocation stage, we have that $\pbRule(\tuple{\projSet, c, b}, \profile) \cap \idealSet_1(\projSet \cup \projSet') \neq \emptyset$. This proves that agent~1's manipulation would also a successful anticipative manipulation, given a completion principle that is top-sufficient.
		
		We can generalise this to all $k > 1$ for the case of U-FSSP-P. To do so, add $3(k - 1)$ agents in groups of 3. Each group can conceive and approve of two new projects. It is easy to check that all the new projects will always be shortlisted, so we are back to the scenario above.
		
		Note that we only talked about U-FSSP-P and U-FSSP-A. The result for U-FSSP-O is immediately derived from Proposition~\ref{prop:FSSPimplications}.
	\end{proof}
	
	We finally consider the case of the median-based shortlisting rules.
	
	\begin{proposition}
		\label{prop:medianShortRuleFSSP}
		Let $\delta$ be the Euclidean distance over $\mathbb{R}^2$, for all $k \in \mathbb{N}_{> 0}$, the pair $\tuple{\medianShortRule_{k, \delta}, \pbRule}$, where $\pbRule$ is a unanimous allocation rule, is neither U-FSSP-P nor U-FSSP-O, for any completion principle that satisfies top-necessity and top-sufficiency.
		
		Moreover, if $\pbRule$ is determined, then the pair $\tuple{\medianShortRule_{1, \delta}, \pbRule}$ is not U-FSSP-A, for any completion principle that satisfies top-sufficiency.
	\end{proposition}
	\begin{proof}
		We first prove the claim for $k = 1$ and then explain how to generalise it to all $k \in \mathbb{N}_{>0}$ (only for U-FSSP-P and U-FSSP-O). Consider the shortlisting instance $I = \tuple{\conceivProjSet, c, b}$ with $\conceivProjSet = \{p_1, \ldots, p_6\}$, all projects having cost~1, and a budget limit of $b = 3$. Suppose the distance~$\delta$ is the usual distance in the plane, with the projects being positioned as in the figure below.
		
		
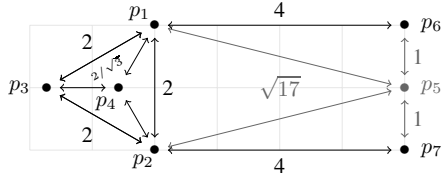
\begin{center}
			\begin{tikzpicture}
				\draw[step=1.0,color=black!10,thin] (-2, -1) grid (4, 1);
				
				\node[label={[label distance=-1em]135:$p_1$}] (d1) at (0, 1) {$\bullet$};
				\node[label={[label distance=-1em]205:$p_2$}] (d2) at (0, -1) {$\bullet$};
				\node[label={[label distance=-0.2em]180:$p_3$}] (d3) at (-1.73, 0) {$\bullet$};
				\node (d4) at (-0.577, 0) {$\bullet$};
				\node at (-0.777, -0.25) {$p_4$};
				\node[label={[label distance=-0.2em, color=black!60]0:$p_5$}, color=black!60] (d5) at (4, 0) {$\bullet$};
				\node[label={[label distance=-0.2em]0:$p_6$}] (d6) at (4, 1) {$\bullet$};
				\node[label={[label distance=-0.2em]0:$p_7$}] (d7) at (4, -1) {$\bullet$};
				
				\path[<->] (d1) edge node[right] {2} (d2);
				\path[<->] (d1) edge node[above left] {2} (d3);
				\path[<->] (d2) edge node[below left] {2} (d3);
				
				\path[<->] (d1) edge node[] {} (d4);
				\path[<->] (d2) edge node[] {} (d4);
				\path[<->] (d3) edge node[] {} (d4);
				\node[rotate = 30] at (-0.77, 0.3) {\tiny $2/\sqrt{3}$};
				
				\path[<->, color=black!60] (d1) edge node[] {} (d5);
				\node[color=black!60] at (2, 0) {$\sqrt{17}$};
				\path[<->] (d1) edge node[above] {4} (d6);
				\path[<->] (d1) edge node[below left] {} (d3);
				\path[<->] (d2) edge node[below] {4} (d7);
				\path[<->] (d2) edge node[above left] {} (d3);
				\path[<->, color=black!60] (d2) edge node[right] {} (d5);
				
				\path[<->, color=black!60] (d5) edge node[right, color=black!60] {1} (d6);
				\path[<->, color=black!60] (d5) edge node[right, color=black!60] {1} (d7);
			\end{tikzpicture}
		\end{center}
		
		We assume that two agents are involved in the process. The first agent is such that $p_1 \rhd_1 p_2 \rhd_1 p_3 \rhd_1 p_5 \rhd_1 \ldots$ and $C_1 = \{p_1, p_2, p_3, p_5\}$. The second agent is such that $p_4 \rhd_2 p_6 \rhd_2 p_7 \rhd_2 \ldots$ and $C_2 = \{p_4, p_6, p_7\}$. The truthful shortlisting profile is then $\shortProfile = (\{p_1, p_2, p_3\}, \{p_4, p_6, p_7\})$. All projects have thus been submitted to the first stage, except for $p_5$. The shortlisting rule $\medianShortRule_{1, \delta}$ would thus consider the following three clusters: $\{p_1, p_2, p_3, p_4\}$, $\{p_6\}$, and $\{p_7\}$, and the shortlist would be $\projSet = \{p_4, p_6, p_7\}$.
		
		Now, assume that agent~1 submits $\{p_1, p_2, p_5\}$ instead of $\{p_1, p_2, p_3\}$, leading to the shortlisting profile $\shortProfile' = (\{p_1, p_2, p_5\}, \{p_4, p_6, p_7\})$. Then, the clusters will be $\{p_1\}$, $\{p_2, p_4\}$, and $\{p_5, p_6, p_7\}$ (for a suitable tie-breaking between Voronoï partitions). The shortlist would then be $\projSet' = \{p_1, p_2, p_5\}$.
		
		Interestingly, we have $\idealSet_1(\projSet \cup \projSet') = \{p_1, p_2, p_5\} = \projSet'$.
		We have thus reached a similar construction as the one in the proof of Proposition~\ref{prop:reprShortRuleFSSP}. Every non-empty budget allocation in $\allocSet(\tuple{\projSet', c, b})$ is strictly preferred by agent~1 to any budget allocation in $\allocSet(\tuple{\projSet, c, b})$. The empty budget allocation is weakly preferred by agent~1 to any budget allocation in $\allocSet(\tuple{\projSet, c, b})$. Given that $\pbRule$ is unanimous, agent~1's manipulation is pessimistically successful. The same applies in the case of U-FSSP-A when $\pbRule$ is determined.
		
		We can extend this to all $k > 1$ for the case of U-FSSP-P. To do so, add $k - 1$ agents, all knowing and approving of three new projects of cost~1. All the new projects are placed uniformly on a circle with centre $p_4$ and a radius large enough so that all new projects will be in their own cluster. Then all the new projects will be shortlisted and we are back to the original case for $k = 1$.
		
		Note that we only talked about U-FSSP-P and U-FSSP-P. The result for U-FSSP-O is immediately derived from Proposition~\ref{prop:FSSPimplications}.
	\end{proof}
	\noindent Note that the previous statement can be generalised to many distances.
	
	\medskip
	
	Interestingly both of the above statements are about unrestricted FSSP, though the successful manipulations presented in the proofs only make use of the awareness set of the manipulator. That is because in both proofs, the manipulator is initially aware of their top projects and manipulates by restraining from submitting some. This hints at some potentially interesting refinements of the FSSP axioms that can be worth studying. For instance, by restraining the type of ballots a manipulator can submit, or said differently, the type of strategic behaviours they can engage into.
	\begin{table}
		\centering
		\setlength\aboverulesep{5pt}
		\setlength\belowrulesep{5pt}
		\newcommand{\spacingrule}{\arrayrulecolor{white}\specialrule{1pt}{1pt}{1pt}\arrayrulecolor{black}}
		\begin{tabular}{rccc}
			\toprule
			& $\tuple{\nominationShortRule, \pbRule}$ & $\tuple{\reprShortRule_k, \pbRule}$ & $\tuple{\medianShortRule_{k, \delta}, \pbRule}$ \\
			
			\midrule
			
			\begin{tabular}{@{}r@{}}R-FSSP-P, R-FSSP-A, \\ and R-FSSP-O\end{tabular} & \xmark{} & \xmark{} {\scriptsize $(k \geq 2)$} & \xmark{} {\scriptsize $(k \geq 2)$} \\ \spacingrule
			\footnotesize Result Statement & \footnotesize Theorem~\ref{thm:Imposibility_R_FSSP} & \footnotesize Theorem~\ref{thm:Imposibility_R_FSSP} & \footnotesize Theorem~\ref{thm:Imposibility_R_FSSP} \\ \spacingrule
			\footnotesize Condition on $\pbRule$ & \footnotesize Determined & \footnotesize Determined & \footnotesize Determined \\ \spacingrule
			\footnotesize \begin{tabular}{@{}r@{}}Condition on the \\Completion Principle\end{tabular} & \footnotesize Top-First & \footnotesize Top-First & \footnotesize Top-First \\
			
			\midrule
			
			U-FSSP-P & \cmark & \xmark{} {\scriptsize $(k \geq 1)$} & \xmark{} {\scriptsize $(k \geq 1)$} \\ \spacingrule
			\footnotesize Result Statement & \footnotesize Proposition~\ref{prop:NomShortRule_UFSSPP} & \footnotesize Proposition~\ref{prop:reprShortRuleFSSP} & \footnotesize Proposition~\ref{prop:medianShortRuleFSSP} \\ \spacingrule
			\footnotesize Condition on $\pbRule$ & \footnotesize Unanimous & \footnotesize \begin{tabular}{@{}c@{}}Unanimous or \\ Determined\end{tabular} & \footnotesize \begin{tabular}{@{}c@{}}Unanimous or \\ Determined\end{tabular} \\ \spacingrule
			\footnotesize \begin{tabular}{@{}r@{}}Condition on the \\Completion Principle\end{tabular} & \footnotesize Top-First & \footnotesize \begin{tabular}{@{}c@{}}Top-Necessity and \\Top-Sufficiency\end{tabular} & \footnotesize \begin{tabular}{@{}c@{}}Top-Necessity and \\Top-Sufficiency\end{tabular} \\
			
			\midrule
			
			U-FSSP-O & \xmark{} & \xmark{} {\scriptsize $(k \geq 1)$} & \xmark{} {\scriptsize $(k \geq 1)$} \\ \spacingrule
			\footnotesize Result Statement & \footnotesize Proposition~\ref{prop:NomShortRule_UFSSPA} & \footnotesize Proposition~\ref{prop:reprShortRuleFSSP} & \footnotesize Proposition~\ref{prop:medianShortRuleFSSP} \\ \spacingrule
			\footnotesize Condition on $\pbRule$ & \footnotesize \begin{tabular}{@{}c@{}}Only for \\ Specific Rules\end{tabular} & \footnotesize \begin{tabular}{@{}c@{}}Unanimous or \\ Determined\end{tabular} & \footnotesize \begin{tabular}{@{}c@{}}Unanimous or \\ Determined\end{tabular} \\ \spacingrule
			\footnotesize \begin{tabular}{@{}r@{}}Condition on the \\Completion Principle\end{tabular} & \footnotesize \footnotesize \begin{tabular}{@{}c@{}}Cost-Neutral \\ Monotonicity\end{tabular} & \footnotesize \begin{tabular}{@{}c@{}}Top-Necessity and \\Top-Sufficiency\end{tabular} & \footnotesize \begin{tabular}{@{}c@{}}Top-Necessity and \\Top-Sufficiency\end{tabular} \\
			
			\midrule
			
			U-FSSP-A & \xmark{} & \xmark{} {\scriptsize $(k = 1)$} & \xmark{} {\scriptsize $(k = 1)$} \\ \spacingrule
			\footnotesize Result Statement & \footnotesize Proposition~\ref{prop:NomShortRule_UFSSPA} & \footnotesize Proposition~\ref{prop:reprShortRuleFSSP} & \footnotesize Proposition~\ref{prop:medianShortRuleFSSP} \\ \spacingrule
			\footnotesize Condition on $\pbRule$ & \footnotesize \begin{tabular}{@{}c@{}}Only for \\ Specific Rules\end{tabular} & \footnotesize Determined & \footnotesize Determined \\ \spacingrule
			\footnotesize \begin{tabular}{@{}r@{}}Condition on the \\Completion Principle\end{tabular} & \footnotesize \begin{tabular}{@{}c@{}}Cost-Neutral \\ Monotonicity\end{tabular} & \footnotesize Top-Sufficiency & \footnotesize Top-Sufficiency \\
			\bottomrule
		\end{tabular}
		\caption{Summary of the results on first-stage strategy-proofness for different shortlisting rules. For each FSSP axiom (or set of axioms), we indicate whether it is satisfied or not in general, and under which conditions on the allocation rule $\pbRule$ and the completion principle the results applies. $\pbRule$ is an arbitrary allocation rule, and $\delta$ is the Euclidean distance over $\mathbb{R}^2$.}
		\label{tab:Summary_FSSP}
	\end{table}

	\section{Conclusion}
	\label{sec:Conclusion}
	
	The aim of this paper was to capture more accurately the different stages that occur in real-life PB processes. Starting from the standard model of PB, we introduced a preliminary stage in which voters can submit proposals that will then be shortlisted to form an allocation instance, \textit{i.e.}, an instance of the standard model. We explored two lines of work within this \emph{end-to-end} model: one that relates to the creation of the shortlist, and another one that relates to the interactions between the two stages.
	
	On the shortlisting side, we presented three different shortlisting rules, tailored to fulfil different objectives for the shortlist. The three main objectives we identified are the following: reducing the number of shortlisted projects, representing all the agents in the shortlist, and, avoiding having overly similar projects in the shortlist. The shortlisting rules $\reprShortRule$ and $\medianShortRule$ performs well on the first objective. They are also good candidates for the second and third objective, respectively. They also enjoy interesting axiomatic properties, as we have seen.
	
	On the interaction side, we focused on the strategic behaviour that can emerge during the shortlisting stage, due to the two-stage nature of the process. Taking our time to reflect on what exactly it would mean for an agent---who has in mind the budget allocation determined in the second stage---to engage in strategic behaviour during the first stage, we introduced six notions of \emph{first-stage strategy-proofness}. For each of the shortlisting rules we introduced, we then successively studied under which conditions they are immune to manipulation, or not. Our results are summarised in Table~\ref{tab:Summary_FSSP}. The main take-home message here is that, overall, shortlisting rules are not immune to manipulation, and that there can always be an incentive for an agent not to propose to build a fountain in the centre of the main square.
	
	Overall, this paper is a first step towards the principled investigation of the full PB process. There are still many other features deserving attention. Other types of interaction between the two stages can be investigated, such as devising allocation rules that take into account not only the outcome of the shortlisting stage but also the shortlisting profile itself, to enforce some kind of fairness across the two stages. This leads to the idea of defining and studying single rules for the whole process instead of taking the composition of a shortlisting and an allocation rule. The shortlisting stage can also be studied in terms of efficiency, or more accurately in terms of the loss of efficiency it entails (by not shortlisting all proposals). Concepts such as the distortion---a measure of the loss of efficiency due to the limited information available---\citep{PrRo06} could be used for that purpose.
	
	\bibliographystyle{ACM-Reference-Format}
	\bibliography{PB}
	
\end{document}